\documentclass[12pt,a4paper]{article}

\usepackage{amsmath, amssymb, amsfonts, amsthm}
\usepackage[cp1251]{inputenc}  
\usepackage[T2A]{fontenc}     

\newtheorem{theorem}{Theorem}[section]

\newtheorem{proposition}[theorem]{Proposition}

\theoremstyle {definition}
\newtheorem{example}[theorem]{Example}
\newtheorem{remark}[theorem]{Remark}

\numberwithin{equation}{section}

\newcommand{\Z}{\mathbb Z}
\newcommand{\Q}{\mathbb Q}
\newcommand{\N}{\mathbb N}

\renewcommand{\:}{\colon}
\renewcommand{\>}{\rightarrow}

\title{Secure cloud computations: \\
Description of (fully)homomorphic ciphers within the P-adic model of encryption}

\author{ Andrei Khrennikov and Ekaterina Yurova\\ International Center for Mathematical Modelling\\
in Physics and Cognitive Sciences\\ 
Linnaeus University, V\"axj\"o, S-35195, Sweden\\
 Andrei.Khrennikov@lnu.se}

\begin{document}
\maketitle

\begin{abstract}
In this paper we consider the description of homomorphic and fully  homomorphic ciphers in the $p$-adic model of encryption. This model describes a wide class of ciphers, but certainly not all. Homomorphic and fully homomorphic ciphers are used to ensure the credibility of remote computing, including cloud technology. The model describes all homomorphic ciphers with respect to arithmetic and coordinate-wise logical operations in the ring of $p$-adic integers $\Z_p$. We show that there are no fully homomorphic ciphers for each pair of the considered set of arithmetic and coordinate-wise logical operations on $\Z_p$. We formulate the problem of constructing a fully homomorphic cipher as follows. We consider a homomorphic cipher  with respect to operation "$*$" on $\Z_p$. Then, we describe the complete set of operations "$G$", for which the cipher is homomorphic. As a result, we construct a fully homomorphic cipher with respect to the operations  "$*$" and "$G$". We give a description of all operations "$G$", for which we obtain fully homomorphic ciphers with respect to the operations "$+$" and "$G$" from 
the homomorphic cipher constructed with respect to the operation "$+$".  We also present examples of such "new" operations.
\end{abstract}

{\bf keywords:} p-adics, p-adic functions, homomorphic ciphers, fully  homomorphic ciphers, dynamical systems, 1-Lipschitz

\section{Introduction}
Cloud computing and storage solutions provide users and enterprises with various capabilities to store and process their data in third-party data centers, \cite{Cloud}.  Homomorphic encryption is a form of encryption that allows computations to be carried out on ciphertext, thus generating an encrypted result which, when decrypted, matches the result of operations performed on the plaintext. A cryptosystem that supports arbitrary computation on ciphertexts is known as fully homomorphic encryption (FHE), \cite{Wiki}. The existence of an efficient and fully homomorphic cryptosystem would have great practical implications in the outsourcing of private computations, for instance, in the context of cloud computing, \cite{CloudHom}.

A brief review of the known homomorphic encryption algorithms is presented in \cite{Obzor_Ind}. Examples of implemented fully homomorphic ciphers exist - see, for example,  \cite{Gentry_3}, \cite{Gentry_2}, \cite{Gentry_1}, \cite{Patent_hom}. \footnote{However, there are some drawbacks in these schemes. For example, in the scheme of C. Gentry, the size of the ciphertext and computational complexity for encryption and decryption grows exponentially depending on the number of operations on the plaintext. In  \cite{Patent_hom}, to compute the encrypted value of the product of plaintexts, it is necessary to use one of the plain texts explicitly.}

The general idea of homomorphic (and fully homomorphic) encryption is as follows (for example, \cite {Obz_appl_1}, \cite {Obz_appl_2}). Suppose we have a  set of data $M$. The operations $g_1 : M\times M \to M,$ $g_2 : M\times M \to M$ are defined on the set $M$ (for homomorphic encryption only one operation is required). It is necessary to find the value of an expression $W(d_1,\ldots,d_n)$, which is defined through the operations $g_1$ and $g_2$ on the data  $d_1,\ldots,d_n\in M.$ By analogy with the formulas of Boolean algebra, the expression $W$ can be considered as a formula in the basis $g_1$ and $g_2$ (for more details see section \ref{sec_Appl}). If the calculation of the formula $W$ is performed remotely (for example, using cloud services), then the user sends the data $d_1,\ldots,d_n$ to an untrusted environment (for example, to the cloud server). After that, the result of the computing returns to the user. In this case, the user's data become open. 

We understand a cipher as a family of bijective transformations of the set $M,$ where each transformation is identified by a certain parameter -- the encryption key. Suppose we have a cipher with the following property with respect to the operations $g_1,$ $g_2$. For every transformation of encryption $f_a,$ where $a$ is a key, the following relations  $f_a(g_1(x,y))=g_1(f_a(x),f_a(y))$ and $f_a(g_2(x,y))=g_2(f_a(x),f_a(y)),$ $x,y\in M$ hold. Then, $f_a(W(d_1,\ldots,d_n))=W(f_a(d_1),\ldots,f_a(d_n)).$ This means that the remote computations are performed on encrypted data $f_a(d_1),\ldots,f_a(d_n)$ and the result of calculations $W(d_1,\ldots,d_n)$ is obtained in encrypted form $f_a(W).$ That is, only the user has access to the data $d_1,\ldots,d_n$. In general, this approach provides complete trust in remote computing.

In this paper, we consider the description of all homomorphic and fully homomorphic ciphers for the $p$-adic model of encryption with respect to a given set of operations, namely, arithmetic ("$+$" and "$\cdot$") and coordinate-wise logical ("$\mathrm{XOR}$" and "$\mathrm{AND}$"), defined on the set of  $p$-adic integers $\Z_p$. This model involves the automata implementation of encryption functions on the set of words in the alphabet $\{0,1,\ldots,p-1\}$. 

We start our paper by recalling definitions that are related to the $p$-adic analysis, as well as introducing the necessary notations.

In section \ref{sec_model} we describe the $p$-adic model of encryption. Moreover, we show that the problem of description of homomorphic (fully homomorphic) ciphers is reduced to the description of the measure-preserving 1-Lipschitz functions $f : \Z_p\to \Z_p$, which defines a homomorphism with respect to a given operation (relatively, to a given pair of operations) on $\Z_p$.

In the $p$-adic model of encryption the Theorems \ref{t_gom_ariff} and \ref{t_gom_coord} describe all homomorphic ciphers with respect to arithmetic operations "$+$", "$\cdot$" and coordinate-wise logical operations "$\mathrm{XOR}$" and "$\mathrm{AND}$" on $\Z_p$. Based on these theorems,  Proposition \ref{Prop_full_hom} shows that a fully homomorphic cipher does not exist with respect to all possible pairwise combinations of the operations from the set $\{``+", ``\cdot", ``\mathrm{XOR}", ``\mathrm{AND}"\}$. In this regard, in section \ref{sec_full_hom_cipher} we discuss the following possibility of constructing a fully homomorphic cipher within the $p$-adic model of encryption. Let us choose a homomorphic cipher with respect to a given operation "$*$" (for example, we take the operation "$+$" as "$*$"). Then we set on $\Z_p$ the new operation $G : \Z_p\times\Z_p \to \Z_p$ such that a given cipher is homomorphic. Operation $G$ can be set as a convergent in the $p$-adic metric power series on $\Z_p$. In this case, the homomorphic cipher is a fully homomorphic cipher with respect to the operations "$*$" and "$G$". Of course, the new operation "$G$" will not have such properties as associativity, commutativity, etc. In Proposition \ref{new_operation} we give a description of all operations "$G$", for which we obtain fully homomorphic ciphers with respect to the operations "$+$" and "$G$" from the homomorphic cipher with respect to the operation "$+$". At the end of the section \ref{sec_full_hom_cipher}, we present examples of such "new" operations.

In section \ref {sec_Appl} we discuss the possibility of remote computing security by using a fully homomorphic encryption on particularly chosen operations for a given computational procedure (that is, the formulas implemented by a computer program).

%---------

\subsection{$P$-adic numbers}

For any prime number $p$ the p-adic norm $|\cdot|_p$ is defined on $\Q$ in the following way. For every nonzero integer $n$ let $ord_p(n)$ be the highest power of $p$ which divides $n$. Then we define $|n|_p=p^{-ord_p (n)}$, $|0|_p=0$ and $|\frac {n}{m}|_p=p^{-ord_p(n)+ord_p(m)}$.

The completion of $\Q$ with respect to the p-adic metric $\rho_p(x,y)=|x-y|_p$  is called the field of p-adic numbers $\Q_p$.
The metric $\rho_p$ satisfies the so-called strong triangle inequality $|x\pm y|_p\le \max{(|x|_p;|y|_p)}$. 
The set $\Z_p=\{x\in\Q_p\colon |x|_p\leq 1\}$ is called the set of p-adic integers. 

Hereinafter, we will consider only the $p$-adic integers.
Every $x\in \Z_p$ can be expanded in canonical form, namely, in the form of a series that converges for the $p$-adic norm:
$x=x_0+px_1+\ldots+p^kx_k+\ldots,\quad x_k\in\{0,1,\ldots,p-1\}, k\ge 0.$

Let $a\in \Z_p$ and $r$ be positive integers. The set $B_{p^{-r}}(a)=\{x\in\Z_p : |x-a|_p\le p^{-r}\}=a+p^r\Z_p$ is a ball of radius  $p^{-r}$ with center $a$.

%----------------------------------
\subsection {$P$-adic functions}
\label {subsec_Lip}

In this paper, we consider functions $f:\Z_p\> \Z_p$, which satisfy the Lipschitz condition with constant 1 (i.e. 1-Lipschitz functions).
Recall that 
$f:\Z_p\> \Z_p$ is a 1-Lipschitz function if
$$  
|f(x)-f(y)|_p\le |x-y|_p,\quad \mbox {for all}\quad  x,y\in\Z_p. 
$$
This condition is equivalent to the following: $x\equiv y \pmod {p^k}$ follows $f(x)\equiv f(y) \pmod {p^{\;k}}$ for all $k\ge 1$.

For all $k\ge 1$ a 1-Lipschitz transformation $f: \Z_p \rightarrow \Z_p$ of the reduced mapping modulo $p^k$ is  
\begin{equation}
\label {f_mod_p_k}
    f_{k-1}: \Z/p^k\Z \rightarrow Z/p^k\Z,\;\;z\mapsto f(z)\pmod {p^k}.
\end{equation}
Mapping $f_{k-1}$ is well defined (the $f_{k-1}$ does not depend on the choice of representative in the ball $z + p^k\Z_p$).

%--------------------------------------
\subsubsection{Van der Put series}
$P$-adic functions can be represented in the form of the van der Put series.
The van der Put series is defined in the following way. Let $f\:\Z_p\>\Z_p$ be a continuous function. Then there exists a unique sequence of $p$-adic coefficients
$B_0,B_1,B_2, \ldots$ such that 
\begin{equation}
\label{vdp} 
f(x)=\sum_{m=0}^{\infty} B_m \chi(m,x)     
\end{equation}
for all $x \in \Z_p.$
Here the characteristic function $\chi(m,x)$ is given by
$$
\chi(m,x)=\begin{cases}
1,&\text{if $|x-m|_p\le p^{-n}$;}\\
0,&\text{otherwise,}
\end{cases}
$$
where $n=0$ if $m=0$, and $n$ is uniquely defined by the inequality 
$p^{n-1}\leq m \leq p^n-1$ otherwise (see Schikhof's book \cite{Schikhof} for detailed presentation of the theory of the van der Put series). 
The number $n$ in the definition of $\chi(m,x)$ has a very natural meaning. It is just the number of digits in a base-$p$ expansion of $m\in \N_0.$ Then
$$
\left\lfloor  \log_p m \right\rfloor 
=
\left(\text{the number of digits in a base-} p \;\text{expansion for} \;m\right)-1,
$$
therefore $n=\left\lfloor  \log_p m \right\rfloor+1$ for all $m\in\N_0$ and $\left\lfloor  \log_p 0 \right\rfloor=0$ (recall that $\left\lfloor \alpha \right\rfloor$ denotes the integral part of $\alpha$).  

The coefficients $B_m$ are related to the values of the function $f$ in the following way. Let $m=m_0+ \ldots +m_{n-2} p^{n-2}+m_{n-1} p^{n-1}$, 
$m_j\in \{0,\ldots ,p-1\}$, $j=0,1,\ldots,n-1$ and $m_{n-1}\neq 0,$ then

\begin{equation}
\label {koef_vdP}
B_m=\begin{cases}
    f(m)-f(m-m_{n-1} p^{n-1}), &\text{if $m\ge p$;}\\
                         f(m), &\text{otherwise}.
\end{cases}
\end{equation}

1-Lipschitz functions $f\: \Z_p \to \Z_p$ in terms of the van der Put series were described in \cite {Schikhof}. We follow Theorem 3.1 \cite {Tfunc} as a convenience for further study. In this theorem, the function $f$ presented via the van der Put series is 1-Lipschitz if and only if $|B_m|_p \leq p^{-\left\lfloor \log_p m \right\rfloor}$ for all $m \ge 0$. Assuming $B_m=p^{\left\lfloor \log_p m\right\rfloor}b_m$, we find that the function $f$ is 1-Lipschitz if and only if it can be represented as
\begin{equation}
\label {1-Lip}
    f(x)=\sum_{m=0}^\infty p^{\left\lfloor \log_p m\right\rfloor} b_m \chi (m,x)
\end{equation}
for suitable $b_m \in \Z_p,$ $m\ge 0.$

\subsubsection {Coordinate representation of 1-Lipschitz functions}
In this section we describe a coordinate representation of $p$-adic functions, see, for example, \cite{YuRecent}.

 We recall that any $p$-adic integer (as an element of the ring $\Z_p$) can be expanded into the series:
\begin{equation*}
x=x_0 +  p\;x_1 + \ldots+ p^k x_k +\ldots, \;\; x_j \in\{0,1,\ldots,p-1\}.
\end{equation*}
Let us denote 
\begin{align}
\label{_x}
    (x)_k&=x_0+x_1p+\ldots +x_{k-1}p^{k-1},\; k\ge 1, \notag \\
    [x]_k&=(x_0,x_1,\ldots ,x_{k-1}),\; k\ge 1.
\end{align}

Let functions $\delta_k (x), k=0,1,2,\ldots$ be $k$-th digit in a base-$p$ expansion of the number $x\in \Z_p,$ i.e. $\delta_k \: \Z_p\rightarrow \left\{0,1,\ldots,p-1\right\},\; \delta_k(x)=x_k.$

Any map $f: \Z_p \to \Z_p$ can be represented in the form:
\begin{equation}
\label{coor1} 
f(x)=\delta_0 (f(x))+ p\delta_1 (f(x))+ \ldots+ p^k \delta_k (f(x))+\ldots.    
\end{equation}

According to Proposition 3.33 in \cite {ANKH}, $f$ is a 1-Lipschitz function if and only if
for every $k \ge 1$ the k-th coordinate function $\delta_k (f (x))$ does not
depend on $\delta_{k+s}(x)$ for all $ s \ge 1 $, i.e. $\delta_k(f(x+p^{k+1}\Z_p))=\delta_k(f(x))$ for all $x\in \{0,1,\ldots,p^{k+1}-1\}$.

Taking into account notation \ref{_x} for $k\ge 0$, we consider the following functions of $p$-valued logic 
$$
\varphi_k\: \underbrace{\{0,\ldots,p-1\}\times\ldots\times\{0,\ldots,p-1\}}_{k+1} \> \{0,\ldots,p-1\},
$$ 
and $\varphi_k: [x]_{k+1} \mapsto \delta_k(f(x)).$ 

Then any 1-Lipschitz function $f\: \Z_p \> \Z_p$ can be represented as

\begin{equation}
\label{coor2}   
 f(x)=f(x_0+\ldots+p^kx_k+\ldots)=\sum_{k=0}^{\infty}p^k\varphi_k(x_0,\ldots,x_k)=
\sum_{k=0}^{\infty}p^k\varphi_k([x]_{k+1}).
\end{equation}

The function $\varphi_k(x_0,\ldots,x_k)$ can be defined by its sub-functions obtained by fixing the first $k$ variables   $(x_0,\ldots, x_{k-1})$.
By \ref {_x}, the sub-function of the function  $\varphi_k(x_0,\ldots,x_k)$ which is obtained by fixing the variables $x_0=a_0,\ldots,x_{k-1}=a_{k-1},$ $a_i\in\{0,\ldots,p-1\},$ is denoted by $\varphi_{k,[a]_k}$, where $a=a_0+pa_1+\ldots+p^{k-1}a_{k-1}$. In these notations, the function $\varphi_k(x_0,\ldots,x_k)$ can be represented as

\begin{equation}
    \varphi_k(x_0,\ldots,x_k)=\sum_{a=0}^{p^k-1}I_{[a]_k}(x_0,\ldots,x_{k-1})\varphi_{k,[a]_k}(x_k),
\end{equation}
where $I_{[a]_k}$ is a characteristic function, i.e.
$$
I_{[a]_k}(x_0,\ldots,x_{k-1})=\begin{cases}
1,&\text{if $(x_0,\ldots,x_{k-1})=[a]_k$;}\\
0,&\text{otherwise.}
\end{cases}.
$$

Thus, we can rewrite the 1-Lipschitz function $f\: \Z_p \> \Z_p$ as
\begin{equation}
\label{coor3} 
    f(x)=\sum_{k=0}^{\infty}p^k\varphi_k(x_0,\ldots,x_k)=\varphi_0+\sum_{k=1}^{\infty} p^k\sum_{a=0}^{p^k-1}I_{[a]_k}([x]_{k-1})\varphi_{k,[a]_k}(x_k).
\end{equation}

 We call the relation \ref {coor3} the coordinate representation of 1-Lipschitz function $f$.

Any function $ \varphi_k $ can be given by a polynomial of the ring of $(k+1)$-variate polynomials  $(\Z/p\Z)[x_0,\ldots,x_k]$, with coefficients from the residue ring $\Z/p\Z$, whose degree in each variable is at most $p-1$. In other words, $ \varphi_k $ is defined as an element of a factor-ring $(\Z/p\Z)[x_0,\ldots,x_k]$ modulo an ideal generated by all polynomials 
$x_i^p-x_i,\; i=0,1,\ldots,k$.

The function $ \varphi_k $  can be represented in the form (expansion by the leading variable  $x_k$):
$$
\varphi_k(x_0,\ldots,x_k)=x_k^{p-1}\alpha_{p-1}(x_0,\ldots,x_{k-1})
+\ldots+\alpha_{0}(x_0,\ldots,x_{k-1}),
$$
where $\alpha_{i}(x_0,\ldots,x_{k-1})\in (\Z/p\Z)[x_0,\ldots,x_{k-1}],\; i=0,1,\ldots,p-1$. In particular, 
$$
\varphi_{k,[a]_k}=x_k^{p-1}\alpha_{p-1}([a]_k)+\ldots+\alpha_{0}([a]_k)
$$
is a polynomial from  $(\Z/p\Z)[x_k]$ ($\alpha_{i}([a]_k)\in \Z/p\Z,\; i=0,1,...,p-1$).

Thus, the coordinate functions $\varphi_{k,[a]_k}$, $\varphi_0$ can be considered as a function of $p$-valued logic and as a transformation of the ring $\Z/p\Z$.

%-------------------

\subsection{$P$-adic dynamics}

Dynamical system theory studies trajectories (orbits), i.e.
sequences of  iterations:
$$x_0,\ x_1=f(x_0),\ldots,x_{i+1}=f(x_i)=f^{(i+1)}(x_0),\ldots ,$$
where $f^{(s)}(x)=\underbrace{f(f(\ldots f(x))\ldots)}_{s}$.

We consider a $p$-adic autonomous dynamical system $\left\langle \Z_p,\mu_p ,f \right\rangle;$ for more details see, for example, \cite {Unif0}, \cite {Unif}, \cite{Erg}, \cite{ErgodHaos}, \cite{MeraJNT}. 
The space $\Z_p$ is equipped with a natural probability measure, namely,
the Haar measure $\mu_p$ normalized so that $\mu_p(\Z_p) = 1.$
Balls $B_{p^{-r}} (a)$ of nonzero
radii constitute the base of the corresponding $\sigma$-algebra of measurable
subsets, $\mu_p(B_{p^{-r}} (a))= p^{-r}.$ The function $f:\Z_p \to \Z_p$ is continuous on $\Z_p.$

A measurable mapping $f\:\Z_p\>\Z_p$ is called measure-preserving if $$\mu_p(f^{-1}(U)) = \mu_p(U)$$ for each measurable subset $U \subset \Z_p.$ 

%--------------
Criteria of measure-preservation for 1-Lipschitz functions are presented in the following theorems.

\begin{theorem}[V. Anashin, \cite {Unif0}, \cite {Unif} or \cite {ANKH}]
\label {cr_mer_Anashin}
A 1-Lipschitz functions $f:\Z_p \to \Z_p$ preserves the measure if and only if $f_k=f\pmod{p^k}$ is bijective on $\Z/p^k\Z$ for any $k=1,2,\ldots.$
\end{theorem}

\begin{theorem}[Theorem 2.1, \cite {MeraJNT}]
\label {cr_mer_vdp}
A 1-Lipschitz function $f:\Z_p\>\Z_p$ represented by the van der Put series \ref{1-Lip}  preserves the measure if and only if 
\begin{enumerate}
\item $\{b_0,b_1,\ldots, b_{p-1}\}$ establish a complete set of residues modulo $p$ \\ (i.e. $f(x)$ is bijective modulo $p$);
\item the elements in the set $\{b_{m+p^k}, b_{m+2p^k},\ldots ,b_{m+(p-1)p^k}\}$ are all nonzero residues modulo $p$ for any $m=0,\ldots, p^k-1,$ $k\ge 2.$
\end{enumerate}
\end{theorem}

\begin{theorem}[Theorem 2.1  \cite {ErgodHaos}]
\label {cr_mer_coord}
A 1-Lipschitz function $f:\Z_p\>\Z_p$ represented in the coordinate form \ref {coor3} preserves the measure if and only if all functions $\varphi_0$ and $\varphi_{k,[a]_k},$ $a\in \{0,1\ldots,p^k-1\},$ $k\ge 1$ are bijective on $\{0,\ldots,p-1\}$.
\end{theorem}

%=====================================

\section{Model}
\label {sec_model}
We consider the following  automata model of encryption. Recall that $\mathfrak{A}=\langle \mathcal{T}, \mathcal{S}, \mathcal{L}, S, L\rangle$ is called an automaton if $\mathcal{T}$ is  a finite set of input symbols (input alphabet); $\mathcal{L}$ is a finite set of output symbols (output alphabet); $\mathcal{S}$ is a set of states (this set is not necessarily finite); $S : \mathcal{T}\times \mathcal{S} \to \mathcal{S}$ is a state
transition function; $L : \mathcal{T}\times \mathcal{S} \to \mathcal{L}$ is an output function. Automaton $\mathfrak{A}$ transforms the input words from the alphabet $\mathcal{T}$ into the output words from the alphabet $\mathcal{L}$ by the following rules:

\begin{enumerate}
    \item Suppose that at some moment of  discrete time  $k,\;k\ge 0,$  the automaton $\mathfrak{A}$ is in the state $s_k\in \mathcal{S}$;
    \item the symbol $t_k\in \mathcal{T}$ is fed to the input of the automaton $\mathfrak{A};$
    \item as a result, the automaton $\mathfrak{A}$ changes its state to $s_{k+1}=S(t_k,s_k)$;
    \item the symbol  $l_{k+1}=L(t_k,s_k)\in \mathcal{L}$ appears on the output of $\mathfrak{A}$.
\end{enumerate}     

Remember that the initial automaton $\mathfrak{A}_{s_0}=\langle \mathcal{T}, \mathcal{S}, \mathcal{L}, S, L, s_0\rangle$ is an automaton $\mathfrak{A};$ where one state $s_0\in \mathcal{S}$ is fixed, $s_0$ is called the initial state. 

Automaton $\mathfrak{A}_{s_0}$ transforms the complete set of words in the alphabet $\mathcal{T}$ into the complete set of words in the alphabet $\mathcal{L};$ see Figure 1.  This transformation we denote by $f_{\mathfrak{A}_{s_0}}$. 

 \begin{center}
\begin{picture}(250,60)
\put(0,35){\vector(1,0){90}}
\put(160,35){\vector(1,0){90}}
\put(0,40){$\ldots,t_k,\ldots ,t_1 ,t_0$}
\put(170,40){$l_0,l_1, \ldots ,l_k,\ldots$}
\thicklines
\put(90,20){\line(0,1){30}}
\put(90,50){\line(1,0){70}}
\put(90,20){\line(1,0){70}}
\put(160,20){\line(0,1){30}}
\put(120,32){$\mathfrak{A}_{s_0}$}
\put(105,0){Figure 1}
\end{picture}
\end{center}

In particular, $f_{\mathfrak{A}_{s_0}}$ maps a set of words of length $n\in \N$ in the alphabet $\mathcal{T}$ - $\mathcal{T}_n$ to the set of words of length $n$ in the alphabet $\mathcal{L}$ - $\mathcal{L}_n$. Such induced mapping is denoted by $f_{\mathfrak{A}_{s_0}, n}$ (i.e., $f_{\mathfrak{A}_{s_0}, n} : \mathcal{T}_n \to \mathcal{L}_n,$ $n\in\N$).

As the input and output alphabets $\mathcal{T}$ and $\mathcal{L},$ we consider the set $\{0,1,\ldots,p-1\}=\mathcal{T}=\mathcal{L}$, $p$ is a prime number. The complete set of words in the alphabet $\{0,1,\ldots,p-1\}$ can be considered as $\Z_p$ (i.e., the number $t_0+pt_1+\ldots+p^kt_k+\ldots$ is associated with the word $(t_0,t_1,\ldots,t_k,\ldots)$). In this case, $f_{\mathfrak{A}_{s_0}}$ is a map $\Z_p$ to itself. 

The paper \cite {An_avt_2} provides the following theorem (here we reformulate the theorem in a convenient form for further discussion).

\begin {theorem} [V. Anashin, Theorem 2.1 \cite {An_avt_2}]
The automaton function $f_{\mathfrak{A}_{s_0}} : \Z_p\to\Z_p$ of the automaton $\mathfrak{A}_{s_0}=\langle \mathcal{T}, \mathcal{S}, \mathcal{L}, S, L, s_0\rangle,$ $\mathcal{T}=\mathcal{L}=\{0,1\ldots,p-1\}$ is 1-Lipschitz.
Conversely, for every 1-Lipschitz function $f : \Z_p\to \Z_p$ there exists an automaton
$\mathfrak{A}_{s_0}=\langle \{0,1\ldots,p-1\}, \mathcal{S}, \{0,1\ldots,p-1\}, S, L, s_0\rangle$ such that $f=\mathfrak{A}_{s_0}$.
\end {theorem} 

We assume that all induced functions $f_{\mathfrak{A}_{s_0}, n}$ are bijective on the set of words of length $n,\;n\ge 1$ (i.e., by the notation \ref {f_mod_p_k}, $f_{\mathfrak{A}_{s_0}, n}=f_{\mathfrak{A}_{s_0}}\pmod {p^n}$ is bijective on $\Z/p^k\Z$). Then, by Theorem \ref {cr_mer_Anashin} 1-Lipschitz function $f_{\mathfrak{A}_{s_0}}$ preserves the measure. In this case, by Proposition 4.26 from \cite {ANKH} the function  $f_{\mathfrak{A}_{s_0}}$ is bijective on  $\Z_p$.

We consider the words in the alphabet $\mathcal{T}=\{0,1,\ldots,p-1\}$ as a plain text, and the words in the alphabet $\mathcal{L}=\{0,1,\ldots,p-1\}$ as a ciphertext, then the family of bijective on $\Z_p$ functions $f_{\mathfrak{A}_{s_0}},$ $s_0\in \mathcal{S}$ defines a cipher, where the initial state $s_0$ is an encryption key.

\begin{remark}
\label {rem_cipher_not_all}
Not all codes are implemented in the framework of this model of encryption. For example, the RSA cryptosystem or block ciphers can not be implemented in our encryption model. This is due to the fact that these ciphers use alphabets where the number of symbols is not a prime.
\end{remark}

Thus, a family of bijective on $\Z_p$ transformations  $f_{\mathfrak{A}_{s_0}},$ $s_0\in \mathcal{S}$, each of which implements an automaton $\mathfrak{A}$, form a cipher  $\mathfrak{C}$. On $\Z_p$ we set an operation "$g$" as the function  $g : \Z_p\times\Z_p \to \Z_p$. We say that the cipher $\mathfrak{C}$ is a homomorphic cipher with respect to the operation "$g$" if for every $s_0\in \mathcal{S}$
\[
f_{\mathfrak{A}_{s_0}}(g(x,y))=g(f_{\mathfrak{A}_{s_0}}(x),f_{\mathfrak{A}_{s_0}}(y)),\;\;x,y\in\Z_p.
\]

Suppose that operations "$g_1$" and "$g_2$" (it is possible to consider more than two operations) are defined on $\Z_p$, that is  $g_1 : \Z_p\times\Z_p \to \Z_p,$ $g_2 : \Z_p\times\Z_p \to \Z_p$. We call a cipher $\mathfrak{C}$ a fully homomorphic with respect to the operations "$g_1$" and "$g_2$", if
\begin{align*}
f_{\mathfrak{A}_{s_0}}(g_1(x,y))&=g_1(f_{\mathfrak{A}_{s_0}}(x),f_{\mathfrak{A}_{s_0}}(y)),\;\;x,y\in\Z_p;\\
f_{\mathfrak{A}_{s_0}}(g_2(x,y))&=g_2(f_{\mathfrak{A}_{s_0}}(x),f_{\mathfrak{A}_{s_0}}(y)),\;\;x,y\in\Z_p.
\end{align*}

Thus, the problem of describing homomorphic (fully homomorphic) ciphers in the p-adic model of encryption is as follows.

Find all functions $f : \Z_p\to \Z_p$ such that
\begin{enumerate}
    \item $f$ is 1-Lipschitz function (i.e., $f$ is realized by some automaton);
    \item $f$ preserves the measure (ensures the uniqueness of decryption);
    \item $f$ defines a homomorphism with respect to a given set of operations on $\Z_p$ (one operation is for a homomorphic cipher, two operations are for a fully homomorphic cipher).
\end{enumerate}

%*************************

\section {Homomorphic ciphers}
\label {sec_hom_cipher}
In this section we give a description of 1-Lipschitz functions $f:\Z_p\>\Z_p$, which preserve the measure and define the homomorphism relative to a binary operation on $\Z_p$, i.e., to some function $g : \Z_p\times \Z_p \to \Z_p$. In this section we shall consider only the following operations:

\begin{enumerate}
    \item ordinary arithmetical operations "$+$" and "$\cdot$" on $\Z_p;$ 
    \item coordinate-wise logical operations "$\mathrm{XOR}$" and "$\mathrm{AND}$". 
These operations are defined on $\Z_p$ as follows. Let $p$-adic numbers $x, y \in \Z_p$ be defined in the canonical form, i.e.,
$$
x=x_0+x_1p+\ldots+x_kp^k+\ldots;\;\;\;\;\;y=y_0+y_1p+\ldots+y_kp^k+\ldots,
$$
where $x_i, y_i\in \{0,1,\ldots,p-1\},$ $i\ge 0$. 
Then the canonical representations of the $p$-adic numbers $x\mathrm{XOR}y$ and $x\mathrm{AND}y$ have the forms

\begin{multline*}
x\mathrm{XOR}y=((x_0+y_0)\bmod p)+((x_1+y_1)\bmod p)p+\ldots\\
+((x_k+y_k)\bmod p)p^k+\ldots;
\end{multline*}

\begin{multline*}
x\mathrm{AND}y=((x_0\cdot y_0)\bmod p)+((x_1\cdot y_1)\bmod p)p+\ldots\\
+((x_k\cdot y_k)\bmod p)p^k+\ldots;
\end{multline*}
\end{enumerate}

Note that the functions that define the homomorphisms with respect to arithmetic operations "$+$" and "$\cdot$"  on the $p$-adic analogue of the field of complex numbers were considered in \cite{Schikhof}. In contrast to this case, we consider the functions that preserve the measure and define the homomorphism on $\Z_p$ for a wider set of binary operations. Full description of measure-preserving, 1-Lipschitz functions, which define homomorphisms for specific operations on $\Z_p,$ are presented in Theorem \ref{t_gom_ariff} (for arithmetic operations) and Theorem \ref{t_gom_coord} (for logical operations). In other words, these theorems provide a description of all  homomorphic ciphers with respect to each of the operations "$+$", "$\cdot$", "$\mathrm{XOR}$" and "$\mathrm{AND}$". Note that functions that equal identically zero or one are not considered in these theorems.

\begin {theorem}[Arithmetic operations]
\label {t_gom_ariff}
Let $f:\Z_p\to \Z_p$ be a 1-Lipschitz function, which differs from zero and unit functions. 

Then
\begin{enumerate}
    \item $f$ defines a homomorphism with respect to the operation "$+$", i.e., \\
$f(x+y)=f(x)+f(y),\; x,y\in \Z_p$ if and only if \\
$f(x)=Ax,\; A\in \Z_p,\; A\ne 0$. 

Such a function preserves the measure if and only if $A\not \equiv 0 \pmod p$;
    \item $f$ defines a homomorphism with respect to the operation "$\cdot$", i.e., \\
$f(x\cdot y)=f(x)\cdot f(y),\; x,y\in \Z_p$ if and only if 

\begin{equation}
\label {fun_umnz}
    f(x)=\begin{cases}
p^k A^k (t_0^s(\bmod p))\cdot(1+p\;t)^a, &\text{if} \;\; x=p^k(t_0+t p),\\
0,                                               &\text{if}\;\; x=0
\end{cases}
\end{equation}
where $t_0\not \equiv 0\pmod p,$ $t\in \Z_p,$ $a, A \in \Z_p,$ $s\in \{1,\ldots,p-1\},$ $k\ge 0.$
 
Such a function preserves the measure if and only if \\
$A\not \equiv 0\pmod p,$ $\text{GCD}\;(s, p-1)=1,$ $a\not \equiv 0\pmod p$.
\end{enumerate} 
\end {theorem}
     
\begin {proof} Let $f$ define a homomorphism with respect to the operation "$+$". Then $f(m)=m\cdot f(1),$ $m\in \Z.$ Let $f(1)=A\in \Z_p,$ $A\ne 0$. Since 1-Lipschitz function $f$ is continuous on $\Z_p$ and $\Z$ is dense in $\Z_p$, then  $f(x)=A\cdot x,$ $x\in \Z_p$. The function $f(x)=A\cdot x$  preserves the measure if and only if $A\not\equiv 0\pmod p$ (see, for example, Lemma 4.41 in \cite {ANKH}). It is clear that the function $f(x)=Ax$ defines a homomorphism with respect to addition on $\Z_p.$

Let us prove the second statement of the theorem. Let $f$ defines a homomorphism with respect to multiplication on $\Z_p$. Since $f$ is distinct from the identity and zero functions (i.e., there are $x,y\in\Z_p$ such that $f(x)\ne 1$ and $f(y)\ne 0$), and $f$ is the homomorphism, then $f(0)=0.$ Moreover, since $f$ is the homomorphism, then $f(1)=1.$ As $f$ is a 1-Lipschitz function, then $f(p)\equiv 0\pmod p$ and $f(1+pt)\equiv 1\pmod p$. 1-Lipschitz function $f$ induces a mapping  $f_0: \Z/p\Z\to \Z/p\Z,$ $z\mapsto f(z)\pmod {p^k}$. Because $f$ is the homomorphism, then $f_0$ is also  homomorphic  on $\Z/p\Z$ with respect to multiplication. The function $f_0$ is a non-zero function (otherwise, $f(1+pt)\equiv f_0(1)\equiv 0 \pmod p$). Then there exists $s\in \{0,1,\ldots,p-2\}$ such that $f_0(z)=z^s$. Let $f(p)=pA,\; A\in \Z_p$. 

Since any $x\in \Z_p, x\ne 0$ can be represented in the form 
$$
x=p^rt_0(1+pt),\; t\in \Z_p,\; t_0\in \{1,2,\ldots,p-1\},\; r\ge 0,
$$
then the function $f$, which defines a homomorphism with respect to multiplication on $\Z_p,$ can be represented in the form
\[
f(x)=\begin{cases}
p^k\cdot A^k\cdot (t_0^s(\bmod p))\cdot\phi(1+pt), &\text{if} \;\; x=p^k(t_0+t\cdot p),\; k\ge 0,\\
0,                                                 &\text{if}\;\; x=0
\end{cases}
\]
where $t_0\not \equiv 0\pmod p,$ $t\in \Z_p,$ $a,A \in \Z_p,$ $s\in \{0,1,\ldots,p-2\}$ (here we consider the function $t_0^s(\bmod p)$ as a mapping of the set $\{1,2,\ldots ,p-1\}$ into itself), and 1-Lipschitz function $\phi : 1+p\Z_p \to 1+p\Z_p$ define a homomorphism with respect to multiplication on $1+p\Z_p$.

Let us find the representation of the function $\phi$. Let $\mathrm{EXP}_p : p\Z_p\to 1+p\Z_p$ be the $p$-adic exponent and $\mathrm {LN}_p : 1+p\Z_p\to p\Z_p$ be the $p$-adic logarithm. We consider the function $g : p\Z_p\to p\Z_p$ such that $g(\tau)=\mathrm {LN}_p(\phi(\mathrm {EXP} _p(\tau)))$. Then, the function $g$ defines a homomorphism with respect to addition on $p\Z_p$:
\begin{multline*}
g(\tau_1+\tau_2)=\mathrm {LN}_p(\phi(\mathrm {EXP}_p(\tau_1+\tau_2)))=\\
=\mathrm {LN}_p(\phi(\mathrm {EXP}_p(\tau_1)\cdot\mathrm {EXP}_p(\tau_2)))=\\
=\mathrm {LN}_p(\phi(\mathrm {EXP}_p(\tau_1))\cdot\phi(\mathrm {EXP}_p(\tau_2)))=\\
=\mathrm {LN}_p(\phi(\mathrm {EXP}_p(\tau_1)))+\mathrm {LN}_p(\phi(\mathrm {EXP}_p(\tau_2)))=\\
=g(\tau_1)+g(\tau_2).
\end{multline*}

Therefore, there exists $a\in \Z_p$ such that $g(\tau)=a\tau$. Since $\mathrm {EXP}_p(\mathrm {LN}_p(1+pz))=1+pz,$ $z\in \Z_p$, then 
$$
\mathrm {EXP}_p(g(\tau))=\mathrm {EXP}_p(a\cdot\tau)=\mathrm {EXP}_p(\tau)^a=\phi(\mathrm {EXP}_p(\tau)).
$$

Let $x=1+pt=\mathrm {EXP}_p(\tau),$ $\tau\in p\Z_p$. Then $\phi(x)=x^a,$ $a\in \Z_p$. 

Thus, the function $f$ can be represented in the form 
$$
f(x)=f(p^kt_0(1+pt))=p^k\cdot A^k\cdot (t_0^s(\bmod p))\cdot(1+pt)^a.
$$

Performing the corresponding calculations, we see that the function of this type defines a homomorphism on $\Z_p$ with respect to multiplication.

Let us find the values $A,a\in \Z_p,$ $s\in \{1,2,\ldots,p-1\}$, where the function $f$ of the form
\ref{fun_umnz} preserves the measure. For this we use the criterion of Theorem \ref{cr_mer_vdp}. Let us find the value of the van der Put coefficients of the function $f$. Let $t\in \{0,1,\ldots,p^r-1\},$ $t_0\in \{1,\ldots,p-1\},$ $h\in \{1,2\ldots,p-1\},$ $k\ge 0$. Then $B_0=f(0)=0$ and 
\begin{multline*}
b_{t_0p^k+tp^{k+1}+p^{k+r}h}=\frac {1}{p^{k+r}}B_{t_0p^k+tp^{k+1}+p^{k+r}h}\equiv \\
\equiv \frac {1}{p^{k+r}}\left(f(t_0p^k+tp^{k+1}+p^{k+r}h)-f(t_0p^k+tp^{k+1})\right)\equiv \\
\equiv a\cdot A^k\cdot(t_0^s(\bmod p))\cdot h\pmod p,\; r\ge 1,
\end{multline*}
\[
b_{t_0p^k}=\frac {1}{p^{k}}B_{t_0p^k}\equiv \frac {1}{p^{k}} f(t_0p^k)\equiv A^k\cdot(t_0^s(\bmod p)) \pmod p,\; r\ge 0.
\]

Since $t_0^s\not \equiv 0(\bmod p)$, then $\{b_{t_0p^k+tp^{k+1}+p^{k+r}h}\;:\; h=1,2,\ldots,p-1\}$ coincides with the set of all non-zero residues modulo $p$ if and only if $a\not\equiv 0\pmod p,$ and $A\not\equiv 0\pmod p$. The set $\{b_{t_0p^k}\;:\; t_0=1,2,\ldots,p-1\}$ coincides with the set of all non-zero residues modulo $p$ as $\text{GCD}\;(s, p-1)=1$. Since $f(t_0+pt)\equiv t_0^s \pmod p,$ $t_0\in \{1,\ldots,p-1\}$ and $f(0)=0$, then $f(\bmod p)$ is bijective on $\Z/p\Z$. Then, by Theorem \ref{cr_mer_vdp} the function $f$ preserves the measure if and only if 
$$a\not \equiv 0 \pmod p;\;\;\;A\not \equiv 0 \pmod p;\;\;\; \text{GCD}\;(s, p-1)=1.$$
\end {proof} 

\begin {remark}
If in \ref{fun_umnz}  we set $a=n,$ $s=n,$ $A=p^{n-1}$ for some $n\in \N$, then $f(x)=x^n.$ That is, all such polynomials define a homomorphism with respect to multiplication on $\Z_p$. Functions of the form $f(x)=x^n$ for $n\ge 1$ do not preserve the measure.   
\end {remark}

\begin {theorem}[Logical operations]
\label {t_gom_coord}

Let $f:\Z_p\to \Z_p$ be a 1-Lipschitz function defined in the coordinate form, i.e.,
\begin{equation*}
f(x)=f(x_0+\ldots+p^kx_k+\ldots)=\sum_{k=0}^{\infty}p^k\varphi_k(x_0,\ldots,x_k),
\end{equation*}
where $\varphi_k(x_0,\ldots,x_k)$ are $p$-valued logical functions.
Then 
\begin{enumerate}
   \item $f$ defines a homomorphism with respect to the operation "$\mathrm{XOR}$" if and only if 
\[
\varphi_k(x_0,\ldots,x_k)=\alpha_{0}^{(k)}x_0+\alpha_{1}^{(k)}x_1+\ldots+\alpha_{k}^{(k)}x_k,
\]
where $\alpha_{i}^{(k)}\in\{0,\ldots,p-1\},$ $0\le i\le k,$ $k\ge 0$.

Such functions preserve the measure if and only if 
$\alpha_{k}^{(k)}\not \equiv 0 \pmod p,$ $k\ge 0;$
   \item $f$ defines a homomorphism with respect to the operation "$\mathrm{AND}$" if and only if  
\[
\varphi_k(x_0,\ldots,x_k)=x_0^{s_{0}^{(k)}}\cdot x_1^{s_{1}^{(k)}}\cdots x_k^{s_{k}^{(k)}},
\]
where $s_{i}^{(k)}\in\{1,\ldots,p-1\},$ $0\le i\le k,$ $k\ge 0$.

Such functions preserve the measure if and only if 
\[
f(x)=f(x_0+px_1+\ldots+p^k x_k+\ldots)=\sum_{k=0}^{\infty}p^k(x_k^{s_{k}^{(k)}}(\bmod p)),
\]
where $\text{GCD}\;(s_{k}^{(k)}, p-1)=1,$ $k\ge 0.$
\end{enumerate} 
\end {theorem} 
 
\begin {proof} Let $f$ define a homomorphism with respect to the operation "$\mathrm{XOR}$" on $\Z_p$, i.e., $\varphi_k(x_0+y_0,\ldots,x_k+y_k)=\varphi_k(x_0,\ldots,x_k)+\varphi_k(y_0,\ldots,y_k),$ $x_i,y_j\in\{0,1,\ldots,p-1\},$ $k\ge 0.$

Let 
\[
\varphi_{k,r}(x)=\varphi_{k}(\underbrace {0,\ldots,0}_{r}, x,0,\ldots,0),\;0\le  r\le k.
\]
Since $\varphi_{k,r}(x+y)=\varphi_{k,r}(x)+\varphi_{k,r}(y),$ $x,y\in \Z/p\Z$, then $\varphi_{k,r}(x)$ is the homomorphism on $\Z/p\Z$ with respect to addition. Therefore, $\varphi_{k,r}(x)=a_r^{(k)}x,$ $a_r\in \Z/p\Z$.
Since
\[
\varphi_k(x_0,\ldots,x_k)=\varphi_{k,0}(x_0)+\ldots+\varphi_{k,k}(x_k),
\]
then $\varphi_k(x_0,\ldots,x_k)=a_0^{(k)}x_0+a_1^{(k)}x_1+\ldots+a_k^{(k)}x_k,$ $k\ge 0,$ where $a_i^{(j)}\in \{0,\ldots,p-1\}.$

It is clear that a function represented by the  coordinate functions  defines a homomorphism on $\Z_p$ with respect to the operation "$\mathrm{XOR}$".

Let 
\[
f(x)=f(x_0+x_1p+\ldots)=\sum_{k=0}^{\infty} ([a_0^{(k)}x_0+a_1^{(k)}x_1+\ldots+a_k^{(k)}x_k ](\bmod p))p^k.
\]
Here, by the notation $(*)(\bmod p)$, we emphasize the fact that $\varphi_k$ is defined by the operations on $\Z/p\Z$. 
Coordinate sub-functions of the function $f$ from the representation \ref{coor3} have the form $c+a_k^{(k)}x_k,$ $c\in \{0,\ldots,p-1\}$. These sub-functions are bijective on $\Z/p\Z$ as $a_k^{(k)}\not\equiv 0\pmod p$. Thus, by Theorem \ref{cr_mer_coord} the function $f$ preserves the measure as soon as $a_k^{(k)}\not\equiv 0\pmod p,$ $k\ge 0$.

Let us prove the second statement of the theorem. Let $f$ be a homomorphism with respect to the operation "$\mathrm{AND}$" on $\Z_p$, i.e., 
\begin{multline*}
\varphi_k(x_0\cdot y_0,\ldots,x_k\cdot y_k)=\\
=\varphi_k(x_0,\ldots,x_k)\cdot \varphi_k(y_0,\ldots,y_k),\; x_i,y_j\in\{0,1,\ldots,p-1\},\; k\ge 0.
\end{multline*}

%We assume that $\varphi_k : \underbrace {\Z/p\Z\times\cdots\times\Z/p\Z}_{k+1} \to \Z/p\Z$.
Let 
\[
\varphi_{k,r}(x)=\varphi_{k}(\underbrace {1,\ldots,1}_{r}, x,1,\ldots,1),\;0\le  r\le k.
\]
Since $\varphi_{k,r}(x\cdot y)=\varphi_{k,r}(x)\cdot\varphi_{k,r}(y),$ $x,y\in \Z/p\Z$, then $\varphi_{k,r}(x)$ is the homomorphism on $\Z/p\Z$ with respect to multiplication. Therefore, $\varphi_{k,r}(x)=x^{s_r^{(k)}},$ $s_r\in \{0,1,\ldots,p-2\}$.
Since
\[
\varphi_k(x_0,\ldots,x_k)=\varphi_{k,0}(x_0)\cdots\varphi_{k,k}(x_k),
\]
then $\varphi_k(x_0,\ldots,x_k)=x_0^{s_0^{(k)}}\cdot x_1^{s_1^{(k)}}\cdots x_k^{s_k^{(k)}},$ $k\ge 0$ ($\varphi_k$ is given through the operations on $\Z/p\Z$).

It is clear that a function represented by the  coordinate functions  defines a homomorphism on $\Z_p$ with respect to the operation  "$\mathrm{AND}$".

Let 
\[
f(x)=f(x_0+x_1p+\ldots)=\sum_{k=0}^{\infty} ([x_0^{s_0^{(k)}}\cdot x_1^{s_1^{(k)}}\cdots x_k^{s_k^{(k)}}](\bmod p))p^k.
\]
Here $\varphi_k$ is also defined by the operations on $\Z/p\Z$. 

Coordinate sub-functions of the function $f$ from the representation \ref{coor3} have the form $a_0^{s_0^{(k)}}\cdot a_1^{s_1^{(k)}}\cdots a_{k-1}^{s_{k-1}^{(k)}}\cdot x_k^{s_k^{(k)}}, a_{i}^{s_{i}}\in \{0,\ldots,p-1\},$ $0\le i\le k-1$. These sub-functions are bijective on $\Z/p\Z$ as soon as $a_0^{s_0^{(k)}}\cdot a_1^{s_1^{(k)}}\cdots a_{k-1}^{s_{k-1}^{(k)}}\equiv 1\pmod p$ for any $a_{i}\in \{0,\ldots,p-1\},$ $0\le i\le k-1$ (this is equivalent to $s_0^{(k)}\equiv s_1^{(k)}\equiv\ldots \equiv  s_{k-1}^{(k)}\equiv 0 \pmod p$) and $\text{GCD}\;(s_k^{(k)}, p-1)=1$. 

Thus, by Theorem \ref{cr_mer_coord} the function $f$ preserves the measure as soon as $s_0^{(k)}\equiv s_1^{(k)}\equiv\ldots \equiv s_{k-1}^{(k)}\equiv 0 \pmod p$ and $\text{GCD}\;(s_k^{(k)}, p-1)=1,\;k\ge 0$.
\end {proof}

%----------------------------------------------

\section {Fully homomorphic ciphers}
\label {sec_full_hom_cipher}
In this section we describe fully homomorphic ciphers with respect to each pair of operations from the set $\{``+", ``\cdot", ``\mathrm{XOR}", ``\mathrm{AND}"\}$. As we show in Proposition \ref{Prop_full_hom}, there are no such fully homomorphic ciphers.

Therefore, we consider the following problem. Suppose there exists a homomorphic cipher with respect to the operation  "$*$". We consider  "$+$" on $\Z_p$ as such an operation (these ciphers were described in Theorem \ref{t_gom_ariff}). We want to find a new operation "$G$" such that this cipher defines a homomorphism with respect to new operation "$G$". All possible new operations for a homomorphic cipher with respect to "$+$" are described in Proposition \ref{new_operation}. As a method of setting a new operation, we consider convergent in the $p$-adic metric power series in two variables.

Let $\mathcal{H}(*)$ be the set of all non-zero and non-trivial 1-Lipschitz functions, which define a homomorphism with respect to the operation "$*$" on $\Z_p$ and preserve the measure. For arithmetic and logical coordinate-wise operations, the sets 
\begin{equation}
\label {clas_hom}
    \mathcal{H}(+),\; \mathcal{H}(\cdot),\; \mathcal{H}(\mathrm{XOR}),\; \mathcal{H}(\mathrm{AND})
\end{equation}
are described in Theorems \ref{t_gom_ariff} and \ref{t_gom_coord}. 

The set of functions, consisting of identical function $f(x)=x$, is denoted by  $I.$

\begin {proposition} In the notation \ref{clas_hom} the following relations hold:
\label {Prop_full_hom}
\begin{enumerate}
    \item $\mathcal{H}(+)\cap\mathcal{H}(\cdot)=I;$
    \item $\mathcal{H}(+)\cap\mathcal{H}(\mathrm{XOR})=I;$
    \item $\mathcal{H}(+)\cap\mathcal{H}(\mathrm{AND})=I;$
    \item $\mathcal{H}(\cdot)\cap\mathcal{H}(\mathrm{XOR})=I;$
    \item $\mathcal{H}(\cdot)\cap\mathcal{H}(\mathrm{AND})=I;$
    \item $\mathcal{H}(\mathrm{XOR})\cap\mathcal{H}(\mathrm{AND})=I;$
\end{enumerate}
\end {proposition} 

\begin {proof} 
Let us prove the first statement of the theorem. Let $f(x)\in \mathcal{H}(+)\cap\mathcal{H}(\cdot)$. As $f(x)\in\mathcal{H}(\cdot)$, then $f(1)=1$. Thus, $f(1)=A=1$, that is, $f(x)=x.$

Let us prove the second statement. Let $f(x)\in \mathcal{H}(+)\cap\mathcal{H}(\mathrm{XOR})$. As $f\in \mathcal{H}(+)$, then $f(x)=Ax,$ $A\in \Z_p$. From $f(x\mathrm{XOR}y)=f(x)\mathrm{XOR}f(y)$ it follows that $A(x\mathrm{XOR}y)=Ax\mathrm{XOR}Ay$. Suppose $x=1$ and $y=p-1$. Then 
\begin{equation}
\label {f_xor}
    A\mathrm{XOR}(p-1)A=0.
\end{equation}

It is clear that this equation holds for $A=1$. Let $A=A_0+A_1p+\ldots,\;A_k\in \{0,\ldots,p-1\}$. Note that
\begin {multline*}
(p-1)A=(p-1)(A_0+A_1p+\ldots)=\\
=[p-A_0]+p[A_0-1+p-A_1]+p^2[A_1-1+p-A_2]+\ldots.
\end {multline*}

We find the values $A_0, A_1, \ldots$, using the equality \ref{f_xor}. We obtain that $p-A_0+A_0\equiv 0 \pmod p$;  $A_0-1+p-A_1+A_1\equiv 0 \pmod p$, i.e., $A_0=1$; $A_1-1+p-A_2+A_2\equiv 0 \pmod p$, so $A_1=1$ and so on. As a result, we find that \ref{f_xor} also holds for $A=1+p+p^2+\ldots=\frac {1}{1-p}$. In this case 
\begin{equation}
\label {f_1_xor}
    f(x)=Ax=A(x_0+px_1+\ldots)=\sum_{k=0}^{\infty}(x_0+x_1+\ldots+x_k)p^k.
\end{equation}

We represent this function in the coordinate form \ref{coor2}. Here,  we denote the operation of addition on $\Z/p\Z$ by "$\oplus$".
Note that relation \ref{f_1_xor} for the function $f$ implies that  $\varphi_0=x_0,$ $\varphi_1(x_0, x_1)=x_0\oplus x_1,$ $x_0+x_1=x_0\oplus x_1+p\delta (x_0,x_1)$, where $\delta (x_0,x_1)=\frac {x_0+x_1-x_0\oplus x_1}{p}$. Since $f\in \mathcal{H}(\mathrm{XOR})$, then   $c_0, c_1, c_2 \in \{0,\ldots,p-1\}$ should exist such that $\varphi_2=c_0x_0\oplus c_1x_1\oplus c_2x_2=x_0\oplus x_1\oplus x_2+\delta (x_0,x_1)$. Since $\delta (0,0)=0$, then $c_2=1$. Suppose $x_2=0,$ $x_1=0$ and $x_2=0,$ $x_0=0$, then $c_0=1,$ $c_1=1$. Then $\frac {x_0+x_1-x_0\oplus x_1}{p}=0$ for $x_2=0$ and $x_0+x_1=x_0\oplus x_1$. 

This contradiction shows that this type of coordinate representation of the function $f$ from \ref{f_1_xor} does not coincide with the representation of the second statement of the Theorem \ref{t_gom_coord} (or $f\not \in \mathcal{H}(\mathrm{XOR})$). Then $f(x)=x$ and $\mathcal{H}(+)\cap\mathcal{H}(\mathrm{XOR})=I$. 

Let us prove statement 3.  Suppose that $f(x)\in \mathcal{H}(+)\cap\mathcal{H}(\mathrm{AND})$, then we obtain $f(p^k)=p^k=Ap^k$. Then $A=1$ and $\mathcal{H}(+)\cap\mathcal{H}(\mathrm{AND})=I$.  

Let us prove statement 4. Let $f\in\mathcal{H}(\cdot)\cap\mathcal{H}(\mathrm{XOR}).$ As $1+p+p^2t=1\mathrm{XOR}(p+p^2t)$, then
\begin {multline}
\label {f_2_xor}
1+pA(1+pt)^a=1\mathrm{XOR}pA(1+pt)=\\
=1\mathrm{XOR}f(p(1+pt))=f(1\mathrm{XOR}(p+p^2t))=\\
=f(1+p+p^2t)=(1+p+p^2t)^a
\end {multline}
Set $t=0$ in \ref{f_2_xor} and differentiate functions from \ref{f_2_xor}, then we get 
\begin{equation}
\label {f_3_xor} 
1+pA=(1+p)^a\;\; \text {and} \;\; A=\left(\frac {1}{1+pt}+p\right)^{a-1}.
\end{equation}

Set $t=0$ in \ref {f_3_xor}, we get $A=(1+p)^{a-1}$ and $1+pA=(1+p)A$. Thus $a=A=1.$ 

Using the representation of the second statement of the Theorem \ref{t_gom_ariff} and the second statement of \ref{t_gom_coord} for the function $f$ for $t\in t_0+p\Z_p,$ $t_0\ne 0,$ we get 
\[
\alpha_{0}^{(0)}t_0\equiv f(t)\equiv t_0^s\pmod p.
\]
Then $\alpha_{0}^{(0)}=1,$ $s=1$ and the function $f$ can be represented in the form 
\[
f(x)=\begin{cases}
p^kt,&\text{if} \;\;x=p^kt,\; t\not \equiv 0 \pmod p,\;k>0;\\
t,   &\text{if} \;\;x=t,\; t\not \equiv 0 \pmod p,\;k=0;\\
0,   &\text{if} \;\;x=0.
\end{cases}
\]
Thus, $f(x)=x$ and $\in \mathcal{H}(+)\cap\mathcal{H}(\mathrm{AND})=I.$

Let us prove statement 5. Let $f\in\mathcal{H}(\cdot)\cap\mathcal{H}(\mathrm{AND}).$ Let $x=x_0+px_1+\ldots+x_kp^k+\ldots,$ $x_i\in\{0,\ldots,p-1\}$. Taking into account the type of representation of functions from $\mathcal{H}(\cdot)$ and $\mathcal{H}(\mathrm{AND})$, we obtain that
$$
p^kx_k^{s_k}=f(x_kp^k)=p^kA^k(x_k^s(\bmod p)).
$$
Then, $s_k=s,$ $k\ge 0,$ $A=1$. Taking into account the representation of $f$ from Theorem \ref{t_gom_ariff} and Theorem \ref{t_gom_coord}, we get
\begin {multline*}
(1+pt)^a=f(1+pt)=f(1+pt_1+\ldots)=\\
=1+p(t_1^s(\bmod p))+\ldots+p^k(t_k^s(\bmod p))+\ldots.
\end {multline*}

Let us find the values of the van der Put coefficients for each representation of the function $f$. Let $1+pt+p^kh=1+pt_1+\ldots+p^{k-1}t_{k-1}+p^kh,$ $h\not\equiv 0 \pmod p$, then
\begin {multline*}
B_{1+pt+p^kh}= (1+pt+p^kh)^a-(1+pt)^a = ah(1+pt)^{a-1}p^k=\\
=(1+\ldots+p^{k-1}(t_{k-1}^s(\bmod p))+p^k(h^s(\bmod p)))-\\
-(1+\ldots+p^{k-1}(t_{k-1}^s(\bmod p)))=h^s(\bmod p)p^k.
\end {multline*}
As $ah(1+pt)^{a-1}p^k=h^s(\bmod p),$ then $a=1,$ $s=1$. Thus,   
\[
f(x)=\begin{cases}
p^kt,&\text{if} \;\;x=p^kt,\; t\not \equiv 0 \pmod p,\;k>0;\\
t,   &\text{if} \;\;x=t,\; t\not \equiv 0 \pmod p,\;k=0;\\
0,   &\text{if} \;\;x=0,
\end{cases}
\]
That is $f(x)=x$ and $\mathcal{H}(\cdot)\cap\mathcal{H}(\mathrm{AND})=I.$

Let us prove statement 6. Let $f\in\mathcal{H}(\mathrm{XOR})\cap\mathcal{H}(\mathrm{AND}).$ Using the coordinate representation of the function  $f$ (Theorem \ref {t_gom_coord}), we obtain
\[
\alpha_{0}^{(k)}x_0+\alpha_{1}^{(k)}x_1+\ldots+\alpha_{k}^{(k)}x_k=x_k^{s_k^{(k)}},\;\; k\ge 0.
\] 
Then, $\alpha_{0}^{(k)}x_0+\alpha_{1}^{(k)}x_1+\ldots+\alpha_{k-1}^{(k)}x_{k-1}$ is a constant as soon as 
\[
\alpha_{0}^{(k)}=\ldots=\alpha_{k-1}^{(k)}=0,\;\; \alpha_{k}^{(k)}=1,\;\;s_k^{(k)}=1.
\]
Thus, $f(x)=x$ and $\mathcal{H}(\mathrm{XOR})\cap\mathcal{H}(\mathrm{AND})=I.$
\end {proof}
%----------------------------------------------------
Proposition \ref{Prop_full_hom} shows that there are no fully homomorphic ciphers for all pairs of arithmetic and coordinate-wise logical operations. However, we may reformulate the problem of finding fully homomorphic ciphers as follows. Let "$*$" be an operation on $\Z_p,$ and $f : \Z_p\to \Z_p$ be a family of 1-Lipschitz functions that preserve the measure. Then we define a homomorphism with respect to this operation. We find  operation "$G$" on $\Z_p$ such that each function $f$ defines a homomorphism with respect to this operation, and that  "$G$" differs from "$*$", i.e., $f(G(x,y))=G(f(x),f(y)),$ $x,y\in \Z_p$. As a result, we find that the function $f$ defines a homomorphism with respect to both operations "$*$" and "$G$", i.e., the family of functions $f$ can be used as a fully homomorphic cipher. As new operation "$G$", we consider an arbitrary mapping $G:\Z_p\times \Z_p \to \Z_p$. Of course, such an operation may not be associative, commutative, etc.

In Proposition \ref{new_operation} we give a description of all operations "$G$" (defined by power series), where linear functions $f(x)=Ax,$ $A\in \Z_p$ define the homomorphism with respect to the operation "$G$". In other words, $f$ is the homomorphism with respect to addition and to new operation "$G$".

We set "$G$" (a function $G:\Z_p\times \Z_p \to \Z_p$) as the convergent power series in $\Z_p\times \Z_p$ (it is sufficient to require that the general term of the series converges to zero in the $p$-adic metric).

The function $G(x,y)$ is given by the convergent power series
\begin{equation}
    \label {op_ful_hom_G}
G(x,y)=c+ax+by+\sum_{k=1}^{\infty}\sum_{i+j=n_k}^{n_k} c_{i,j}x^iy^j,\;\; c_{i,j}, a,b,c\in \Z_p,
\end{equation}
where for any $n_k\in\{n_1, n_2,\ldots \;|\; n_k\in \N,\; 1<n_1<n_2<\ldots\}=\mathcal{N}_G$ there exists $0 \le  i,j \le n_k$ such that $c_{i,j}\ne 0$, and if $n\not \in \mathcal{N}_G$, then $c_{i,j}=0$ for any $0 \le  i,j \le n_k,$ $i+j=n$.

\begin {proposition}
\label {new_operation}
Let $f : \Z_p \to \Z_p,$ $f(x)=Ax,$ $A\in \Z_p,$ $A\ne 0$. The function $f$ defines a homomorphism  with respect to the operation "$G$", given as the series \ref {op_ful_hom_G}  if and only if for $\mathcal{N}_G\ne \emptyset:$
\begin{enumerate}
    \item $c=0;$
    \item $n_k=dq_k+1,\;k\ge 1,$ where $d=\text {GCD} (n_1-1,n_2-1,\ldots,n_k-1,\ldots)$, $q_k\in \N;$
    \item $A^d=1$  
\end{enumerate}
and $G=ax+by$ for any $A\ne 0$ for $\mathcal{N}_G = \emptyset.$
\end {proposition}

\begin {proof} 
Let $f=Ax$ define a homomorphism with respect to the operation "$G$", i.e., $AG(x,y)=G(Ax,Ay)$ and $\mathcal{N}_G\ne \emptyset$. Using the representation \ref{op_ful_hom_G}, we get that  $A$ satisfies the system of equations 
$$
A^{n_1}=A,\;A^{n_2}=A,\;\ldots,\; A^{n_k}=A,\; \ldots
$$
or $A^{n_k-1}=1,$ $k\ge 1.$ This system of equations is equivalent to the equation $A^d=1$, where $d=\text {GCD} (n_1-1,n_2-1,\ldots,n_k-1,\ldots)$.  
Let $x=0,$ $y=0,$ then $Ac=c$ and $c=0$ (since $A\ne 0$). It is easy to see that under the conditions of the proposition, the function $f(x)=Ax$  defines homomorphism with respect to the operation "$G$".
If $\mathcal{N}_G=\emptyset$, then $G(x,y)=c+ax+ay$ and $c=0.$  The function $f(x)=Ax$ defines the homomorphism with respect to $G=ax+by$ for any $A$.
\end {proof}

\begin {remark}
Here we recall known facts about the number of solutions of the equation $A^d=1$ in $\Z_p$. If $p\nmid d$, then the equation $A^d=1$ has $\text {GCD} (d,p-1)$ solutions in $\Z_p$ (see, for example, Theorem 3.24 in \cite {ANKH}). If $d=p^k$, then the equation $A^{p^k}=1$ has a unique solution $A=1,$ except when $p=2$ and $k=1$ (in this case, the equation $A^2=1$ in $\Z_2$ has solutions $A=1,$ $A=-1;$ see, for example, Theorem 18.9 in \cite {Katok}). Clearly, for $d=p^kn,$ $p\nmid d$ the equation $A^d=1$ has $\text {GCD} (n,p-1)$ solutions in $\Z_p$.

Thus, for $\mathcal{N}_G\ne \emptyset,$  maximum number of functions of the form $f(x)=Ax$, which define the homomorphism with respect to the operation "$G$", equals $p-1.$ If the functions $f(x)=Ax$ are chosen for the construction of fully homomorphic ciphers with respect to operations "$+$" and "$G$", the prime number $p$ must be large. 
\end {remark}

\begin{example}
Here are some examples of operations "$G$", for which the functions $f(x)=Ax$ define the homomorphism with respect to the operations "$+$" and "$G$" for a suitable choice of $A$:

\begin{enumerate}
    \item $G(x,y)=ax+by,$ $a,b\in \Z_p$ for any $A\in \Z_p,$ $A\ne 0;$
    \item $G_1(x,y)=xy^{p-1}$ or $G_2(x,y)=x^{p-1}y+xy^{p-1}$ for $A\in\Z_p,$ $A^{p-1}=1$. The operation $G_2$ is commutative;
    \item for $p\ne 2,$ then $G(x,y)=x^{\frac{p-1}{2}}\cdot y^{\frac{p-1}{2}},$ $A^{p-1}=1;$
    \item $G(x,y)=\sum_{s=0}^{\infty}p^s\left(x^{(p-1)s+1}+y^{(p-1)s+1}\right)=\frac {x}{1-px^{p-1}}+\frac {y}{1-py^{p-1}},$ and $A^{p-1}=1$. 
\end{enumerate}
\end{example}

%--------------------------------

\section {Application}
\label {sec_Appl}
In Proposition \ref{Prop_full_hom} we have shown that a fully homomorphic cipher does not exist (within the p-adic model of encryption) for each pair of operations from the set $\{``+", ``\cdot", ``\mathrm{XOR}", ``\mathrm{AND}"\}$. Suppose that, using the pairs of operations from this set, we can write an equation that can be implemented by a computer program. From this perspective, the result of Proposition \ref{Prop_full_hom} means that within the $p$-adic model of encryption there are no "universal" (i.e., for any software) fully homomorphic ciphers. However, as was shown in Proposition \ref{new_operation}, it is possible to construct a fully homomorphic cipher with respect to a pair of "new" operations on $\Z_p$. In this section we will discuss such a possibility.

We assume that in the framework of the $p$-adic model of encryption a computer program implements a formula of the original data. This formula is written using the set of operations $g_1:\Z_p\times \Z_p \to \Z_p$ and $g_2:\Z_p\times \Z_p \to \Z_p$. By analogy with the formulas of Boolean algebra, let us define formulas with respect to the operations $g_1$ and $g_2$ over $\Z_p$:
\begin{enumerate}
    \item  variables and operations $g_1,\;g_2$ are formulas;
    \item if $F_1,\;F_2$ are formulas, then $g_1(F_1,F_2),\; g_2(F_1,F_2)$ are formulas.    
\end{enumerate}

We denote the set of all formulas defined with respect to operations $g_1$ and $g_2$ as $[g_1,g_2]$. For instance, $x^3+y+y^2\cdot z$ is the formula from the set $[``+",``\cdot"].$ The following assertion holds.

\begin {proposition}
\label {formuls}

Let operations $g_1:\Z_p\times \Z_p \to \Z_p$ and $g_2:\Z_p\times \Z_p \to \Z_p$ be defined by the formulas from $[``+",``\cdot"].$ Let 1-Lipschitz function $f: \Z_p\to \Z_p$ define a nontrivial homomorphism with respect to the operations $g_1$ and $g_2.$

Then $[g_1,g_2]\subset [``+",``\cdot"]$ and $[g_1,g_2]\ne [``+",``\cdot"].$ 
\end {proposition}

\begin {proof}
Since $g_1, g_2 \in [``+",``\cdot"],$ then $[g_1,g_2]\subset [``+",``\cdot"].$ Assume that $[g_1,g_2]=[``+",``\cdot"].$ Then "$+$" and "$\cdot$" are defined by the formulas $\Psi_{``+"}(x_1,x_2)$ and $\Psi_{``\cdot"}(x_1,x_2)$ with respect to the operations $g_1$ and $g_2.$ Since $f$ is a homomorphism with respect to $g_1$ and $g_2,$ then  

\begin{eqnarray*}
    f(x_1+x_2)&=&f(\Psi_{``+"}(x_1,x_2))=\Psi_{``+"}(f(x_1),f(x_2))=f(x_1)+f(x_2);\\
    f(x_1\cdot x_2)&=&f(\Psi_{``\cdot"}(x_1,x_2))=\Psi_{``\cdot"}(f(x_1),f(x_2))=f(x_1)\cdot f(x_2);
\end{eqnarray*}
i.e. $f$ is the homomorphism with respect to "$+$" and "$\cdot$". Then from Proposition \ref{Prop_full_hom} it follows that $f$ is an identity mapping. This contradicts the condition of the Proposition.
\end {proof}

Proposition \ref {formuls} shows that if  a fully homomorphic cipher exists with respect to certain operations $g_1$ and $g_2$ (formulas with respect to arithmetic operations  "$+$" and "$\cdot$"), then the cipher is not applicable for all formulas from $[``+",``\cdot"].$ Previously, we  assumed the class $[``+",``\cdot"]$ was "universal" in the sense that any computer program could implement the formula from this class.

In this regard, we propose the following method using a fully homomorphic encryption to secure remote computing within the $p$-adic model of encryption.

Let $W(d_1,\ldots,d_n)$ be a given formula (or program), by which cloud computing will be performed. Here $d_1,\ldots,d_n \in \Z_p$ stand for data. This formula can be given, for example, in the basis of the usual arithmetic operations. Let us find a new pair of operations $g_1$ and $g_2,$  such that:
\begin{enumerate}
    \item the formula $W\in [g_1, g_1];$
    \item fully homomorphic cipher $f_a$ with respect to the operations $g_1$ and $g_2$ exists, where $a$ is the key and $[g_1, g_1]\ne [``+",``\cdot"]$.
 \end{enumerate}
Then $W(f_a(d_1),\ldots,f_a(d_n))=f_a(W(d_1,\ldots,d_n)).$ As a result, we are able to perform cloud computing in secure mode for a given formula (program) $W.$ 

\begin{example}
Let 
$$
W(x,y,z)=x^{p-1}y^{(p-1)^2}z+x^{p-1}y^{p-1}z+x^py^{p(p-1)}\left(1+ y^{p^2-3p+2}\right) 
$$
be a formula by which cloud computing will be performed. 

Let us consider the operation of usual addition "$+$" in $\Z_p$, and the operation "$*$" in $\Z_p$ defined by $a*b=ab^{p-1}.$ Then
$$
W(x,y,z)=z*(x*y)+(z*x)*y+(x*x)*(y*y)+x*((x*y)*y).
$$

From Proposition \ref{Prop_full_hom} it follows that the functions $f:\Z_p\to \Z_p,$ $f(x)=Ax,$ $A\in \Z_p,$ $A\ne 0$ and $A^{p-1}=1$ define the homomorphisms with respect to operations "$+$" and "$*$". In other words, the family of such functions is a fully homomorphic cipher with respect to the given operations. In particular, $f(W(x,y,z))=W(f(x),f(y),f(z)).$
\end{example}

\end{document}